\documentclass[11pt]{article}
\usepackage[a4paper,top=3cm,bottom=3.5cm,left=2.5cm,right=2.5cm,marginparwidth=1.75cm]{geometry}
\usepackage[T1]{fontenc}
\usepackage{graphicx}
\usepackage[title]{appendix}
\usepackage{tikz}
\usepackage{pgfplots}
\usepackage{cutwin}

\usepackage{amsmath,yhmath}
\usepackage{amsthm}
\usepackage{amssymb}
\usepackage{hyperref}
\usepackage{cleveref}
\usepackage{epstopdf}
\usepackage{comment}
\usepackage{todonotes}
\usepackage{color}
\usepackage{float}
\usepackage[sort&compress,numbers]{natbib}
\usetikzlibrary{shapes.geometric}
\usetikzlibrary{arrows,automata,quotes}
\usepackage{iftex}
\usepackage{xfrac} 
\usepackage{amsmath,yhmath}
\usepackage{amsthm}
\usepackage{amssymb}
\usepackage{mathtools}
\usepackage{nccmath}
\usepackage{hyperref}
\usepackage{caption, subcaption}
\usepackage[title]{appendix}
\usepackage{cleveref}
\usepackage{epstopdf}
\usepackage{comment}
\usepackage{todonotes}
\usepackage{color}
\usepackage{float}
\usepackage{cutwin}
\usepackage[sort&compress,numbers]{natbib}
\usepackage{makecell}
\usepackage{array}   % for \newcolumntype macro
\usepackage{multirow}
\usepackage{graphicx}
\usepackage{caption}
\usepackage{subcaption}
\newtheorem{thm}{Theorem}[section]
\newtheorem{cor}[thm]{Corollary}
\newtheorem{lemma}[thm]{Lemma}
\newtheorem{prop}[thm]{Proposition}
\theoremstyle{definition}
\newtheorem{definition}{Definition}[section]

\newtheorem{rem}[definition]{Remark}

\newcommand{\inc}{\mathrm{in}}
\newcommand{\s}{\mathrm{sc}}
\newcommand{\p}{\partial}
\newcommand{\ds}{\displaystyle}

\newcommand{\nm}{\noalign{\smallskip}}

\newcommand{\R}{\mathbb{R}}

\newcommand{\Z}{\mathbb{Z}}

\newcommand{\svdots}{\raisebox{3pt}{\scalebox{.6}{$\vdots$}}}
\newcommand{\sddots}{\raisebox{3pt}{\scalebox{.6}{$\ddots$}}}

\numberwithin{equation}{section}
\numberwithin{figure}{section}
\title{Scattering from time-modulated subwavelength resonators\thanks{\footnotesize
This work was supported in part by the Swiss National Science Foundation grant number
200021--200307.}}
\author{Habib Ammari\thanks{\footnotesize Department of Mathematics, ETH Z\"urich, R\"amistrasse 101, CH-8092 Z\"urich, Switzerland (habib.ammari@math.ethz.ch, jinghao.cao@sam.math.ethz.ch, liora.rueff@sam.math.ethz.ch).} \and Jinghao Cao\footnotemark[2] \and Erik Orvehed Hiltunen\thanks{\footnotesize Department of Mathematics, Yale University, New Haven, Connecticut, USA (erik.hiltunen@yale.edu).} \and Liora Rueff\footnotemark[2] }
\date{}
\begin{document}
\maketitle
\begin{abstract}
	We consider wave scattering from a system of highly contrasting resonators with time-modulated material parameters. In this setting, the wave equation reduces to a system of coupled Helmholtz equations that models the scattering problem. We consider the one-dimensional setting. In order to understand the energy of the system, we prove a novel higher-order discrete, capacitance matrix approximation of the subwavelength resonant quasifrequencies. Further, we perform numerical experiments to support and illustrate our analytical results and show how periodically time-dependent material parameters affect the scattered wave field. 
\end{abstract}
\noindent{\textbf{Mathematics Subject Classification (MSC2000):} 35J05, 35C20, 35P20, 74J20}
\vspace{0.2cm}

\section{Introduction}
Time-modulated subwavelength metamaterials have recently emerged as a new paradigm for wave manipulation at subwavelength scales \cite{erik_JCP,reviewTM1,reviewTM2,reviewTM3}. The exploration of time as a novel degree of freedom for metamaterials has led to a plethora of wave phenomena such as nonreciprocity, zero refractive index, and wave amplification \cite{jinghao_liora,jinghao1,jinghao2,erik_JCP}. The building blocks of subwavelength metamaterials are subwavelength resonators: objects exhibiting resonant phenomena in response to wavelengths much greater than their size. The highly contrasting material parameters (relative to the background medium) of these objects are the crucial mechanism responsible for their subwavelength response. We now have a complete understanding of different wave scattering phenomena arising from static subwavelength resonator structures \cite{aihp,Hai_Habib,feppon_1d,florian_habib_ss_hc,effective_bubbles,sini1,sini2,sini3,jun}. However, little is known when waves are scattered by time-modulated subwavelength resonators. The main difficulty for studying wave scattering from time-modulated subwavelength resonators arises from folding of quasifrequencies into the first Brillouin zone in time. Hence, we shall only consider the quasifrequencies corresponding to eigenmodes essentially supported in the subwavelength regime and make the assumption that all the other eigenmodes do not contribute significantly to the scattering field away from the resonator system \cite{erik_JCP}.\par 
Our aim in this paper is to provide a mathematical and computational framework to describe wave scattering from systems of time-modulated high-contrast resonators. To ensure a strong coupling between the incident wave and the resonators, we assume that the incident  frequency is of the same order as the frequency of modulation of the material parameters. In this low-frequency  regime, we establish an expression of the scattered field in one space dimension. To complete the analytical formulation of the eigenmodes arising from a system of time-modulated resonators, we solve the governing equations numerically. \par
Furthermore, we introduce a new notion of the total energy of a system of time-modulated resonators (see Definition \ref{def:totEnergy}). We show how the total energy is affected by the time-modulated material parameters inside the resonators. It has been proven that one-dimensional static materials obey energy conservation \cite{feppon_1d}. Unlike the static case, we will show that the energy is not conserved in a time-modulated material. In order to perform accurate numerical experiments, we prove a novel higher-order discrete, capacitance matrix approximation, which can be applied to a single-resonator system and a general system of $N>1$ resonators. \par
Our theory in this paper applies to subwavelength resonator systems where the modulation frequency is of the same order as the subwavelength resonant quasifrequencies and the operating frequency, \textit{i.e.} the frequency of the incident wave, is near the resonant quasifrequencies. In this regime, strong interactions between the time-modulated resonators and the incident wave occur. Our results can also be generalised to higher space dimensions. Furthermore, it is worth emphasising that the time-modulations considered in this paper are very different from the travelling wave-form modulations discussed in, for instance, \cite{pendry1,pendry2,engheta,liu,rizza,nassar,horsley1,horsley2,Sharabi:22,Pendry:21}. By exploiting the subwavelength resonant properties of the system, they lead to fundamentally different wave propagation and scattering phenomena. Moreover, their mathematical treatment is more involved.\par 
Our paper is structured as follows. We start by providing an overview of the problem setting and introduce the governing equations in the form of a wave equation with time-modulated material parameters and suitable transmission conditions in \hyperref[sec:chpt2]{Section 2}. We particularly assume the material parameters to be periodically time-modulated. We focus on the one-dimensional setting for which we refer to our previous work \cite{jinghao_liora} for an explicit formulation of the solution. 
In \hyperref[sec:chpt3]{Section 3} we prove relevant theorems which we shall use in the following section. Specifically, we derive high-order estimates of the subwavelength resonant quasifrequencies in terms of the contrast parameter $\delta$. Our derivations are based on new capacitance matrix formulations. They apply to the case of a single resonator as well as to the case of multiple resonators. In \hyperref[sec:chpt4]{Section 4} we use that to formulate the scattered field from a system of subwavelength resonators. In \hyperref[sec:chpt5]{Section 5} we introduce the notion of energy for time-modulated systems and complement it with numerical simulations. We summarise our results in \hyperref[sec:chpt6]{Section 6} and formulate some extensions of our findings.

\section{Problem setting}\label{sec:chpt2}
We consider the wave propagation through a medium $\mathcal{U}\subset\mathbb{R}$, composed of high contrast materials with time-dependent material parameters. Let $\{D_i\}_{1\leq i\leq N}$ be a disjoint collection of $N$ resonators distributed inside $\mathcal{U}$. We denote the region taken up by the resonators by
\begin{align*}
    D:=\bigcup_{i=1}^ND_i.
\end{align*}\par 
In the one-dimensional case considered in this paper, we will follow the notation introduced in \cite{jinghao_liora}. In particular, we denote the resonators by $D_i:=\left(x_i^-,x_i^+\right)$, where $\left(x_i^{\pm}\right)_{1\leq i\leq N}$ are the $2N$ boundary points satisfying $x_i^+<x_{i+1}^-$, for any $1\leq i\leq N-1$. Moreover, the length of each resonator is defined by $\ell_i:=x_i^+-x_i^-$ and the length of the gap between the $i$-th and the $(i+1)$-th resonator by $\ell_{i(i+1)}:=x_{i+1}^--x_i^+$.\par 
We denote by $\rho$ and $\kappa$ the system's material parameters. In the acoustic setting we refer to $\rho$ as the material's density and  to $\kappa$ as the material's bulk modulus. We assume that $\rho$ is time-independent and piecewise constant, namely,
\begin{equation*}
    \rho(x)=\begin{cases} 
    \rho_0, & x \notin {D}, \\
    \rho_{\mathrm{r}}, & x \in D_i.
    \end{cases}
\end{equation*}
We assume the material parameter distribution $\kappa$ within each resonator to be periodic in the temporal coordinate $t$ with period $T:=2\pi/\Omega$, namely,
\begin{equation*}
    \kappa(x,t)=\begin{cases} 
    \kappa_0, & x \notin {D}, \\
    \kappa_{\mathrm{r}} \kappa_i(t), & x \in D_i.
    \end{cases}
\end{equation*}
In our numerical simulations, we shall consider $\kappa_i$ given by
\begin{equation}
    \kappa_i(t):=\frac{1}{1+\varepsilon_{\kappa,i}\cos\left(\Omega t+\phi_{\kappa,i}\right)}, \label{eq:rho_kappa}
\end{equation}
for all $1\leq i\leq N$, where  $\varepsilon_{\kappa,i} \in [0,1)$ are the amplitudes of the time-modulations and $\phi_{\kappa,i} \in [0,2\pi)$ the phase shifts. Note that we do not assume $\rho$ to be time-modulated due to the fact that time-modulating $\rho$ does not affect the resonant quasifrequencies at leading order in $\delta$, see \cite{jinghao_liora}.\par
We define the contrast parameter and the wave speeds by
\begin{equation} \label{defdelta}
    \delta:=\frac{\rho_{\mathrm{{r}}}}{\rho_0},\quad v_0:=\sqrt{\frac{\kappa_0}{\rho_0}},\quad v_{\mathrm{r}}:=\sqrt{\frac{\kappa_\mathrm{r}}{\rho_{\mathrm{r}}}},
\end{equation}
respectively. 
We consider an incident wave field $u^{\inc}(x,t)$ with a frequency $\omega$ and such that 
\begin{align}\label{eq:PDE_uin}
    \frac{\partial^2}{\partial t^2} u^{\inc}-v_0^2\Delta u^{\inc}=0.
\end{align}
Then, we let $u^{\mathrm{sc}}$ be defined  by
\begin{equation} \label{defscat}
    u^{\mathrm{sc}}(x,t) = \begin{cases} 
    u(x,t)- u^{\mathrm{in}}(x,t), & x \notin D,\\
    u(x,t), & x \in D,
    \end{cases}
\end{equation}
where $u$ is the total wave. For $x\notin D$, $u^{\mathrm{sc}}$ is the scattered wave. \par 
We assume that $u^{\mathrm{sc}}(x,t) \mathrm{e}^{\mathrm{i} \omega t}$ and $u^{\mathrm{in}}(x,t) \mathrm{e}^{\mathrm{i} \omega t}$ are periodic with respect to $t$ with period $T$ and write their Fourier coefficients with respect to $t$ as follows:  
\begin{equation}\label{def:us_Fourier}
    u^{\mathrm{sc}}(x,t)=\sum\limits_{n=-\infty}^{\infty}v^{\mathrm{sc}}_n(x)\mathrm{e}^{-\mathrm{i}(\omega+n\Omega )t},\quad u^{\mathrm{in}}(x,t)=\sum\limits_{n=-\infty}^{\infty}v^{\mathrm{in}}_n(x)\mathrm{e}^{-\mathrm{i}(\omega+n\Omega )t},\quad\forall\,x\in\mathbb{R},\,t\geq0,
\end{equation}
where $\omega$ is the frequency of the incident wave.\par 
The governing equations for the wave scattering by the collection of resonators $D$ are given by \cite{jinghao_liora,erik_JCP}:
\begin{equation}\label{eq:dD_system_u} 
	\left\{
	\begin{array} {ll}
	\ds \frac{\p^2}{\p t^2}u^\s(x,t) - v_0^2\Delta u^\s(x,t) = 0, \quad &x\notin D, \\[1em]
	\ds \frac{\p }{\p t } \frac{1}{\kappa_i(t)} \frac{\p}{\p t}u^\s(x,t) - v_{\mathrm{r}}^2 \Delta u^\s(x,t) = 0, \quad &x\in D_i, \quad i=1,\ldots,N,\\[1em]
	\ds u^\s|_-(x_i^-,t) - u^\s|_+(x_i^-,t) = u^\inc|_-(x_i^-,t), \quad &i=1,\dots,N, \\[0.5em] 	
    \ds u^\s|_-(x_i^+,t) - u^\s|_+(x_i^+,t) = -u^{\inc}|_+(x_i^+,t), \quad &i=1,\dots,N, \\[0.5em] 
	\ds \frac{\p u^\s}{\p x }\bigg|_{+}(x_i^-,t) - \delta\frac{\p u^\s}{\p x }\bigg|_{-}(x_i^-,t) = \delta \frac{\p u^\inc}{\p x }\bigg|_{-}(x_i^-,t), &i=1,\dots,N, \\[1em]
	\ds  \frac{\p u^\s}{\p x }\bigg|_{-}(x_i^+,t) - \delta\frac{\p u^\s}{\p x }\bigg|_{+}(x_i^+,t) = \delta \frac{\p u^\inc}{\p x }\bigg|_{+}(x_i^+,t), &i=1,\dots,N, \\ [1em]
	u^\s \text{ is an outgoing wave,}
\end{array}		
\right.
\end{equation}
where we use the notation
\begin{align}
    \left.w\right|_{\pm}(x):=\lim_{s\rightarrow0,\,s>0}w(x\pm s).
\end{align}\par 
The subwavelength resonant quasifrequencies $\omega_i,\, i=1,\ldots,N$, are given by the set of $\omega$ such that there is a non-zero solution to \eqref{eq:lin_system} with $u^\textrm{in}(x,t) = 0$. We use the term \textit{quasi}frequency  to emphasise that the real part of $\omega_i$ is defined modulo $\Omega$; here, we choose real parts in the interval $[-\Omega/2,\Omega/2)$. 

We assume that the sizes of the resonators $D_i$ are of order one, the material contrast parameter 
$\delta \ll 1$ and $\Omega = O(\sqrt{\delta})$.  Under these assumptions, we know from \cite{jinghao_liora,erik_JCP} that (\ref{eq:dD_system_u}) with 
$u^\inc =0$ admits $2N$ subwavelength resonant quasifrequencies $\omega_i, i=1,\ldots,N$. These subwavelength quasifrequencies go to zero as ${\delta}$ goes to zero and correspond to eigenmodes essentially supported in the low-frequency regime. They can be computed either by solving a root-finding problem through Muller's method or as the Floquet exponents of a system of $N$ differential equations written in terms of the so-called capacitance matrix. We shall adhere to the latter method.\par 
Decomposing $u^{\mathrm{sc}}$ into its Fourier series as given by (\ref{def:us_Fourier}) allows us to pose (\ref{eq:dD_system_u}) on the modes $v^{\mathrm{sc}}_n$, for all $n\in\mathbb{Z}$, as follows:
\begin{equation}\label{eq:1DL_system}
	\left\{
	\begin{array} {ll}
	\ds  \frac{\mathrm{d}^2}{\mathrm{d}x^2}v^{\mathrm{sc}}_n+\frac{\rho_0(\omega+n\Omega)^2}{\kappa_0}v^{\mathrm{sc}}_n=0, \quad &x\notin D, \\[1em]
	\ds \frac{\mathrm{d}^2}{\mathrm{d}x^2}v_{n}^{\mathrm{sc}}+\frac{\rho_{\mathrm{r}}(\omega+n\Omega)^2}{\kappa_{\mathrm{r}}}v_{i,n}^{**}=0, \quad &x\in D_i, \quad i=1,\ldots,N,\\[1em]
	\ds \left.v^{\mathrm{sc}}_n\right|_{-}\left(x_i^-\right) -\left.v^{\mathrm{sc}}_n\right|_{+}\left(x_i^-\right)=v_n^{\inc}|_-(x_i^-), \quad &\forall\,i=1,\dots,N, \\[0.5em]
    \ds \left.v^{\mathrm{sc}}_n\right|_{-}\left(x_i^+\right) -\left.v^{\mathrm{sc}}_n\right|_{+}\left(x_i^+\right)=-v_n^{\inc}|_+(x_i^+), \quad &\forall\,i=1,\dots,N, \\[0.5em] 	
	\ds \left.\frac{\mathrm{d} v_{n}^{\mathrm{sc}}}{\mathrm{~d} x}\right|_{+}\left(x_i^{-}\right)-\left.\delta \frac{\mathrm{d} v^{\mathrm{sc}}_n}{\mathrm{d} x}\right|_{-}\left(x_i^{-}\right)=\left.\delta\frac{\mathrm{d}v_n^{\mathrm{in}}}{\mathrm{d}x}\right|_{-}(x_i^{-}), &  \forall\,i=1,\ldots,N, \\[1em]
    \ds \left.\frac{\mathrm{d} v_{n}^{\mathrm{sc}}}{\mathrm{~d} x}\right|_{-}\left(x_i^{+}\right)-\left.\delta \frac{\mathrm{d} v^{\mathrm{sc}}_n}{\mathrm{d} x}\right|_{+}\left(x_i^{+}\right)=\left.\delta\frac{\mathrm{d}v_n^{\mathrm{in}}}{\mathrm{d}x}\right|_{+}(x_i^{+}), &  \forall\,i=1,\ldots,N, \\[1em]
	\ds  \left( \frac{\mathrm{d}}{\mathrm{d} |x|} - \mathrm{i}
 \frac{(\omega + n \Omega)}  {v_0}  \right) v^{\mathrm{sc}}_n = 0, & 
 x \in (-\infty, x_1^-) \cup (x_N^+, + \infty).
\end{array}		
\right.
\end{equation}
We denote the wave number corresponding to the $n$-th mode inside and outside of the resonators by
\begin{align*}
    k^{(n)}_{\mathrm{r}}:=\frac{\omega+n\Omega}{v_{\mathrm{r}}},\quad k^{(n)}:=\frac{\omega+n\Omega}{v_0},
\end{align*} 
respectively, while the functions $v_{i, n}^{* *}(x)$ are defined in each resonator $D_i$, $i=1,\dots,N$, through the convolution
\begin{equation}\label{def:conv_v}
    v_{i,n}^{* *}(x)=\frac{1}{\omega+n \Omega} \sum_{m=-\infty}^{\infty} k_{i,m}(\omega+(n-m) \Omega) v^{\mathrm{sc}}_{n-m}(x),
\end{equation}
where $k_{i, m}$ are the Fourier series coefficients of $1 / \kappa_i(t)$:
\begin{align*}
    \frac{1}{\kappa_i(t)} = \sum_{n=-M}^M k_{i, n}\mathrm{e}^{-\mathrm{i}n\Omega t}.
\end{align*}

It has been shown in \cite[Lemma 2.1]{feppon_1d} that the solution of the Helmholtz equation has the following exponential form in $\mathbb{R}\backslash D$:
\begin{align}\label{eq:vi_ansatz_exterior}   v^{\mathrm{sc}}_n(x)=\alpha_n^i\mathrm{e}^{\mathrm{i}k^{(n)}x}+\beta_n^i\mathrm{e}^{-\mathrm{i}k^{(n)}x},\quad\forall\,x\in\left(x_i^+,x_{i+1}^-\right),
\end{align}
for all $i=1,\dots,N-1$. Moreover, the following lemma characterising the solution to the Helmholtz equation inside the resonators has been proved in \cite{jinghao_liora}.
\begin{lemma}
    For each resonator $D_i$, for $i=1,\dots,N$, the interior problem can be written as an infinitely-dimensional system of ordinary differential equations, $\Delta\mathbf{v}_i+C_i\mathbf{v}_i=\mathbf{0}$, with $\mathbf{v}_i:=\begin{bmatrix}v_{i,j}\end{bmatrix}_{j\in\mathbb{Z}}, v_{i,j}=\left.v^{\mathrm{sc}}_j\right|_{D_i}$ and
    \begin{align*}
        C_i:=\left[\begin{smallmatrix}
            & & & & & & &\\
            \sddots & & \sddots & & \sddots & & &\\
            & \gamma^{-K}_{i,-M}(\omega) & \cdots & \gamma^{-K}_{i,0}(\omega) & \cdots & \gamma^{-K}_{i,M}(\omega) & & \\
            & \sddots & & \sddots & & \sddots & & & \\
            & & \gamma^0_{i,-M}(\omega) & \cdots & \gamma^0_{i,0}(\omega) & \cdots & \gamma^0_{i,M}(\omega) &  \\
            & & \sddots & & \sddots & & \sddots & \\
            & & & \gamma^{K}_{i,-M}(\omega) & \cdots & \gamma^{K}_{i,0}(\omega) & \cdots & \gamma^{K}_{i,M}(\omega) \\
            & & & \sddots & & \sddots & & \sddots \\
            & & & & & & &
        \end{smallmatrix}\right],
    \end{align*}
    where the coefficients $\gamma_{i,m}^n(\omega)$ are given by
    \begin{align*}
        \gamma_{i,m}^n(\omega):=k_{i,m}k_{\mathrm{r}}^{(n)}k_{\mathrm{r}}^{(n-m)},\quad \forall -\infty<m,n<\infty.
    \end{align*}
    Let $\{\Tilde{\lambda}_j^i\}_{j\in\mathbb{Z}}$ be the set of all eigenvalues of $C_i$ with corresponding eigenvectors $\{\mathbf{f}^{j,i}\}_{j\in\mathbb{Z}}$. Using the square-roots $\pm\lambda_j^i$ of the eigenvalues $\Tilde{\lambda}_j^i$, the solution to the interior problem over $D_i$ takes the form
    \begin{align}\label{eq:vi_ansatz}
        \mathbf{v}_i(x)=\sum\limits_{j=-\infty}^{\infty}\left(a_j^i\mathrm{e}^{\mathrm{i}\lambda_j^ix}+b_j^i\mathrm{e}^{-\mathrm{i}\lambda_j^ix}\right)\mathbf{f}^{j,i},\quad \forall\,x\in\left(x_i^-,x_i^+\right),
    \end{align}
    for coefficients $\{(a_j^i,b_j^i)\}_{j\in\mathbb{Z}}\subset\mathbb{R}^2$.
\end{lemma}
\begin{rem}\label{rmk:K1}
    As in \cite{jinghao_liora}, we introduce a truncation parameter $K$ enabling us to solve the governing equations (\ref{eq:1DL_system}) numerically for the interior coefficients. Namely, we now write
    \begin{align}\label{eq:vi_ansatz_K}
        \mathbf{v}_i=\sum\limits_{j=-K}^{K}\left(a_j^i\mathrm{e}^{\mathrm{i}\lambda_j^ix}+b_j^i\mathrm{e}^{-\mathrm{i}\lambda_j^ix}\right)\mathbf{f}^{j,i},\quad \forall\,x\in\left(x_i^-,x_i^+\right).
    \end{align}
    This assumption is feasible due to the fact that
    \begin{align*}
        ||v_n||_2=o\left(\frac{1}{n}\right)\quad\text{as}\,\,n\to\infty.
    \end{align*}
\end{rem}
This completes the geometric setup and solution form of our one-dimensional system.

\section{Capacitance matrix approximations}
\label{sec:chpt3}
In this section, we derive two novel approximation formulae to the subwavelength resonant quasifrequencies of (\ref{eq:1DL_system}). Firstly, we introduce a higher-order (in terms of the contrast parameter $\delta$) capacitance matrix approximation for a general $N$ in order to obtain the imaginary parts of the subwavelength resonant quasifrequencies, which are not captured by the already known capacitance matrix approximation \cite{jinghao_liora}. In the single-resonator case, this formula simplifies and admits explicit solutions.\par 
As reference, we recall the static case and a single resonator of length $\ell$. From \cite[Remark 3.4]{feppon_1d}, the subwavelength resonant frequencies are given by
\begin{align}\label{eq:om1_formula}
    \omega_0(\delta)=0,\quad \omega_1(\delta)= - \frac{\mathrm{i} v_{\mathrm{r}}}{\ell} \log \left( 1 + \frac{2 v_{\mathrm{r}} \delta}{ v_0 - v_{\mathrm{r}} \delta}\right)= - \frac{ 2 \mathrm{i} (v_{\mathrm{r}})^2 \delta}{\ell v_0} + O(\delta^3). 
\end{align}
Our numerical results shown in Figure \ref{fig:w_N1} agree with the here presented formula approximating the non-zero subwavelength resonant frequency.
\begin{figure}[H]
    \centering
    \includegraphics[width=0.6\textwidth]{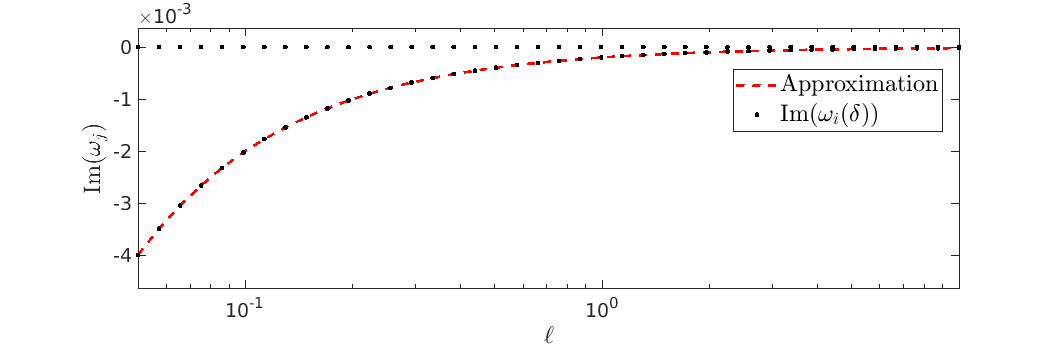}
    \caption{The imaginary parts of the subwavelength resonant frequencies $\omega_0$ and $\omega_1$ for fixed $\delta=0.0001,\,K=4$ and static material parameters, as a function of $\ell$. Note that the real parts of $\omega_0$ and $\omega_1$ are both zero. We use a logarithmic $x$-axis to present our numerical results.}\label{fig:w_N1}
\end{figure}

We now turn to the time-modulated case. To understand the behaviour of the imaginary parts of the subwavelength resonant quasifrequencies of a system of $N$ time-modulated resonators we need to derive a higher-order capacitance matrix approximation for the $N$-resonator system.\par 
Let $D_i =(x_i^-, x_i^+)$. Then for any smooth function $f:\mathbb{R}\rightarrow \mathbb{R}$ we define $ I_{\partial D_i}[f]$ by \cite{jinghao_liora}
\begin{equation} \label{def:i}
    I_{\partial D_i}[f]:= \frac{\mathrm{d} f}{\mathrm{d} x}\bigg\vert_-(x_i^-)- \frac{\mathrm{d} f}{\mathrm{d} x}\bigg\vert_+(x_i^+).
\end{equation}
Next, we define the capacitance matrix $C=(C_{ij})_{i,j=1,\ldots,N}$ by
\begin{align} \label{def:cap}
	C_{ij}:=I_{\partial D_i}[V_j]= \frac{\mathrm{d} V_j}{\mathrm{d} x}\bigg\vert_-(x_i^-)- \frac{\mathrm{d} V_j}{\mathrm{d} x}\bigg\vert_+(x_i^+),
\end{align}
for $ i,j=1,\ldots, N,$ where $V_j$ is the solution of 
\begin{align} \label{def:vi}
        \begin{cases}
            -\frac{\mathrm{d}^2}{\mathrm{d}x^2}V_j=0,&\text{in }\mathbb{R}\backslash D,\\
            V_j(x)=\delta_{ij},&x\in D_i, \\
            V_j(x)= O(1), &x\to \infty.
        \end{cases}
    \end{align}
Hence, we have the following characterisation of the subwavelength resonant quasifrequencies. 
\begin{thm}\label{thm3.2}
   Consider a system of $N\geq 1$ resonators. Let $C$ denote the corresponding capacitance matrix given by (\ref{def:cap}). As $\delta\to0$, the subwavelength resonant quasifrequencies $\omega$ of (\ref{eq:1DL_system}) are, to order $O(\delta^{3/2})$, the quasifrequencies of the system of ordinary differential equations
    \begin{equation}\label{eq:new_CapApprox}
        C\mathbf{c}(t) +\frac{1}{v_0}D\frac{\mathrm{d}}{\mathrm{d}t}\mathbf{c}(t) = -L\frac{ \rho_{\mathrm{r}}}{\delta \kappa_{\mathrm{r}}} \frac{\mathrm{d}}{\mathrm{d} t}\left(K(t)\frac{\mathrm{d} }{\mathrm{d} t}\mathbf{c}(t)\right),
    \end{equation}
    where $L=\mathrm{diag}(\{\ell_i\}_{i=1,\dots,N})$ and $K(t)=\mathrm{diag}(\{\frac{1}{\kappa_i(t)}\}_{i=1,\dots,N})$. The only non-zero entries of the $N\times N$ diagonal matrix $D$ are given by
    \begin{equation*}
    \begin{cases}
        D_{11} = D_{NN} = 1, \quad & \text{if} \ N\geq 2,\\
        D_{11} = 2, & \text{if} \ N = 1.
    \end{cases}
    \end{equation*}
\end{thm}
\begin{rem}
    The above ordinary differential equation can be reformulated as
    \begin{equation*}
        \begin{bmatrix}
        0 & I_N\\
        \nm
        C & \ds \frac{\rho_{\mathrm{r}}}{\delta \kappa_{\mathrm{r}}}L\frac{\mathrm{d}}{\mathrm{d}t}K(t)
        \end{bmatrix} \begin{bmatrix}
            \mathbf{c}(t)\\ 
            \nm
           \ds \frac{\mathrm{d}}{\mathrm{d}t}\mathbf{c}(t)
        \end{bmatrix}
        =  \begin{bmatrix}
        I_N & 0\\
        \nm
        \ds -\frac{1}{v_0}D & \ds -\frac{\rho_{\mathrm{r}}}{\delta \kappa_{\mathrm{r}}}LK(t)
        \end{bmatrix} \begin{bmatrix}
           \ds  \frac{\mathrm{d}}{\mathrm{d}t} \mathbf{c}(t)\\ 
            \nm \ds \frac{\mathrm{d}^2}{\mathrm{d}t^2}\mathbf{c}(t)
        \end{bmatrix},
    \end{equation*}
    where $I_N$ is the $N\times N$ identity matrix.
\end{rem}
\begin{proof}    
    From \cite{jinghao_liora}, we start by claiming that for all $i=1,\dots,N,$ 
    \begin{equation}
    \label{eq:constantvn}
        v_n^{\mathrm{sc}}(x)\big\vert_{[x_i^-,x_i^+]} = c_{i,n} +O(\delta^{1-\gamma}).
    \end{equation}
    We know that
    \begin{equation*}
        v_n^{\mathrm{sc}}(x)\big\vert_{(x_i^+,x_{i+1}^-)} = \alpha_n^i \mathrm{e}^{\mathrm{i}k^{(n)}x}+\beta_n^i\mathrm{e}^{-\mathrm{i}k^{(n)}x}.
    \end{equation*}
    From \cite{feppon_1d} we know that the function $V_i$ defined by 
    (\ref{def:vi}) obeys
    \begin{equation*}
        V_i(x)\big\vert_{[x_i^+,x_{i+1}^-]} =-\frac{1}{\ell_{i(i+1)}} x +\frac{x_{i+1}^-}{\ell_{i(i+1)}},\ \ V_i(x)\big\vert_{D_i} = 1
    \end{equation*}
    and
    \begin{equation*}
        V_{i+1}(x)\big\vert_{[x_i^+,x_{i+1}^-]} =\frac{1}{\ell_{i(i+1)}} x -\frac{x_i^+}{\ell_{i(i+1)}},\ \  V_{i+1}(x)\big\vert_{D_{i+1}} = 1.
    \end{equation*}
    Moreover, from \cite{feppon_1d}, in the intervals $[x_i^+,x_{i+1}^-]$ we have
    \begin{equation*}
        v_n^{\mathrm{sc}}(x)\big\vert_{[x_i^+,x_{i+1}^-]} = c_{i,n} V_i(x) + c_{i+1,n}V_{i+1}(x)+O(\delta).
    \end{equation*}
    In the exterior intervals $(-\infty,x_1^-]$ and $[x_N^+,\infty)$, we use the radiation and  continuity conditions to find that
    \begin{align}\label{eq:vn_N_simpl}
        v_n^{\mathrm{sc}}(x)=\begin{cases}
            c_{1,n}\mathrm{e}^{\mathrm{i}k^{(n)}(x_1^--x)},&\forall\,x\in(-\infty,x_1^-],\\
            c_{N,n}\mathrm{e}^{\mathrm{i}k^{(n)}(x-x_N^+)},&\forall\,x\in [x_N^+,\infty).        \end{cases}
    \end{align}
    Combining these expressions over the whole domain $x\in \R$, we obtain the following expression, up to an error $O\left(\delta\right)$,
    \begin{equation*}
         v_n^{\mathrm{sc}}(x) = \chi_{(x_1^-,x_N^+)}(x)\sum_{j=1 }^{N} c_{j,n}V_j(x)+c_{1,n}\chi_{(-\infty,x_1^-]}(x)\mathrm{e}^{\mathrm{i}k^{(n)}(x_1^--x)}+c_{N,n}\chi_{[x_N^+,\infty)}(x)\mathrm{e}^{\mathrm{i}k^{(n)}(x-x_N^+)},
    \end{equation*}
    where $\chi_A$ is the characteristic function over the set $A$. Applying the operator $I_{\partial D_i}$ defined by (\ref{def:i}), we have
    \begin{equation}
    \label{eq:I_vn_4}
    \begin{split}
        I_{\partial D_i}[v_n^{\mathrm{sc}}] & = \sum_{j=1}^N c_{j,n} C_{ij} -\mathrm{i}k^{(n)}\left( \delta_{i,1}c_{1,n} + \delta_{i,N}c_{N,n}\right)
    \\ &= \frac{\rho_{\mathrm{r}}(\omega+n \Omega)^2}{\delta \kappa_{\mathrm{r}}} \int_{x_i^{-}}^{x_i^{+}} v_{i, n}^{* *}(x)\, \mathrm{d} x.
    \end{split}
    \end{equation}
    Recall that $k^{(n)} = \frac{\omega +n\Omega}{v_0}$ and $v_{i, n}^{* *}(x)=\frac{1}{\omega+n \Omega} \sum_{m=-\infty}^{\infty} k_{i, m}(\omega+(n-m) \Omega) v_{n-m}(x)$. Hence, equation (\ref{eq:I_vn_4}) can be written as
    \begin{equation}\label{eq:pre_caphot}
        \sum_{j=1}^N C_{ij} c_{j,n} -\mathrm{i}k^{(n)}\left( \delta_{i,1}c_{1,n} + \delta_{i,N}c_{N,n}\right) = \frac{\ell_{i}\rho_{\mathrm{r}} (\omega+n\Omega)}{\delta \kappa_{\mathrm{r}}}  \sum_{m=-M}^M k_{i,m} (\omega+(n-m)\Omega)c_{i,n-m},
    \end{equation}
    where $C\in\mathbb{R}^{N\times N}$ is the capacitance matrix introduced in \eqref{def:cap}. Now, define 
    \begin{equation*}
        c_i(t) = \sum_{n=-\infty}^\infty c_{i,n} \mathrm{e}^{-\mathrm{i}(\omega + n\Omega )t}.
    \end{equation*}
    We can compute
    \begin{equation*}
        \begin{split}
            \frac{\mathrm{d}}{\mathrm{d}t}c_i(t) & = -\sum_{n=-\infty}^\infty \mathrm{i}(\omega+n\Omega) c_{i,n}\mathrm{e}^{-\mathrm{i}(\omega+n\Omega )t}, \\
            \frac{\mathrm{d}^2}{\mathrm{d}t^2}c_i(t) & = -\sum_{n=-\infty}^\infty (\omega+n\Omega)^2 c_{i,n}\mathrm{e}^{-\mathrm{i}(\omega+n\Omega )t}.
        \end{split}
    \end{equation*}
    Furthermore, we have
    \begin{equation*}
    \begin{split}
        \frac{\mathrm{d}}{\mathrm{d}t}\left(\frac{1}{\kappa_i(t)}\frac{\mathrm{d}c_i(t)}{\mathrm{d}t}\right) & = \frac{\mathrm{d}}{\mathrm{d}t}\left(\left(\sum_{m=-M}^M k_{i,m}\mathrm{e}^{-\mathrm{i}m\Omega t} \right)\left(-\sum_{n=-\infty} ^\infty \mathrm{i}(\omega+n\Omega) c_{i,n}\mathrm{e}^{-\mathrm{i}(\omega+n\Omega )t}\right)\right)\\
        & = -\frac{\mathrm{d}}{\mathrm{d}t}\sum_{n=-\infty}^\infty \sum_{m=-M}^M \mathrm{i} k_{{i},m} (\omega+n\Omega) c_{i,n} \mathrm{e}^{-\mathrm{i}(\omega+(m+n)\Omega)t} \\
        & = -\sum_{n=-\infty}^\infty \sum_{m=-M}^M k_{{i},m} (\omega+n\Omega)(\omega+(n+m)\Omega) c_{i,n} \mathrm{e}^{-\mathrm{i}(\omega+(n+m)\Omega)t}\\ 
        & = -(\omega+n\Omega)\sum_{n=-\infty}^\infty \sum_{m=-M}^M k_{{i},m} (\omega+(n-m)\Omega)c_{i,n-m} \mathrm{e}^{-\mathrm{i}(\omega+n\Omega)t}.
    \end{split}
    \end{equation*}
    Thus, (\ref{eq:pre_caphot}) implies the desired ordinary differential equation for $\mathbf{c}(t):=\begin{bmatrix}c_i(t)\end{bmatrix}_{i=1}^N$. 
\end{proof}

For a single resonator of length $\ell$, \textit{i.e.} $N=1$, the capacitance matrix is scalar-valued and vanishes: $C=0$. Equation \eqref{eq:new_CapApprox} thereby reduces to
\begin{align}\label{eq:cap_1D_N1}
    \frac{2}{v_0}\frac{\mathrm{d}}{\mathrm{d}t}c(t)+\frac{\ell}{\delta(v_{\mathrm{r}})^2}\frac{\mathrm{d}}{\mathrm{d}t}\frac{1}{\kappa(t)}\frac{\mathrm{d}}{\mathrm{d}t}c(t)=0.
\end{align}
The ordinary differential equation \eqref{eq:cap_1D_N1} admits a closed-form solution. Specifically, we have the following result.
\begin{cor}\label{cor:N1}
    For $N=1$ and as $\delta\to 0$ the non-zero subwavelength resonant quasifrequency $\omega_1$ of (\ref{eq:1DL_system}) is given by 
    \begin{equation*}
        \omega_1 = \frac{\mathcal{K}}{T}\int_0^T {\kappa(t)} \,\mathrm{d}t+O(\delta^{3/2}),
    \end{equation*}
    where we recall that $T=2\pi/\Omega$ and introduce the imaginary constant
    \begin{align*}
        \mathcal{K} := -\frac{2\mathrm{i}(v_{\mathrm{r}})^2\delta}{\ell v_0}.
    \end{align*}
\end{cor}
\begin{proof}
    By reformulating \eqref{eq:cap_1D_N1}, we obtain
    \begin{equation} \label{eq:ct}
        -a(t) \frac{\mathrm{d}}{\mathrm{d}t} c(t) +\frac{\mathrm{d}^2}{\mathrm{d}t^2} c(t) = 0
    \end{equation}
    with
    \begin{equation}
        a(t) := \left(\mathrm{i}\mathcal{K} + \frac{\mathrm{d}}{\mathrm{d}t} \frac{1}{\kappa(t)}\right) \kappa(t).
    \end{equation}
    For the Floquet exponent $\omega_1$ of the non-constant solution $c(t)$ of (\ref{eq:ct}) the following holds:
    \begin{equation*}
        \mathrm{e}^{\mathrm{i}\omega_1 T}= \ds \mathrm{e}^{\int_0^T a(t)\, \mathrm{d}t}.
    \end{equation*}
    Since $\kappa(t)$ is $T$-periodic, the second term of $a(t)$ vanishes in the integral over $(0,T)$ and we obtain
    \begin{equation*}
        \int_0^T a(t)\,\mathrm{d}t = \mathrm{i}\mathcal{K} \int_0^T \kappa(t) \,\mathrm{d}t. 
    \end{equation*}
    This concludes the proof.
\end{proof}
\begin{rem}
    In the case when 
    \begin{equation}
        \kappa(t) = \frac{1}{1+\varepsilon_{\kappa} \mathrm{cos}(\Omega t)},
    \end{equation}
    we can compute the integral explicitly. Namely, we have for $0 \leq \varepsilon_{\kappa} < 1$:
    \begin{equation*}
        \int_0^T \frac{1}{1+\varepsilon_{\kappa} \mathrm{cos}(\Omega t)} \,\mathrm{d}t = \frac{1}{\Omega}\int_0^{2\pi} \frac{1}{1+\varepsilon_{\kappa} \mathrm{cos}(x)}\,\mathrm{d}x= \frac{T}{\sqrt{1-(\varepsilon_{\kappa})^2}}.
    \end{equation*}
    Here we have used that, for $0\leq a<1$,
    \begin{equation*}
        \int_0^{2\pi} \frac{1}{1+a \mathrm{cos}(x)}\,\mathrm{d}x  = \frac{2\pi}{\sqrt{1-a^2}}.
    \end{equation*}
    To conclude, the non-zero resonant quasifrequency of a single resonator is given by \begin{equation}\label{eq:om1_Cepsk}
        \omega_1(\delta)=\frac{\mathcal{K}}{\sqrt{1-(\varepsilon_{\kappa})^2}}+O(\delta^{3/2}).
    \end{equation}

\end{rem}
    
We note that assuming $\kappa(t)=1$, \Cref{cor:N1} yields (\ref{eq:om1_formula}). Next, we verify numerically that in the single resonator case, the system attains a zero resonance $\omega_0(\delta)=0$ and a purely imaginary resonance $\omega_1(\delta)\in\mathrm{i}\mathbb{R}$ which, as shown in  Figure \ref{fig:w_N1_timemod}, decreases for increasing $\varepsilon_{\kappa}$. 
\begin{figure}[H]
    \centering
    %\vspace{0.5cm}
    \includegraphics[width=0.68\textwidth]{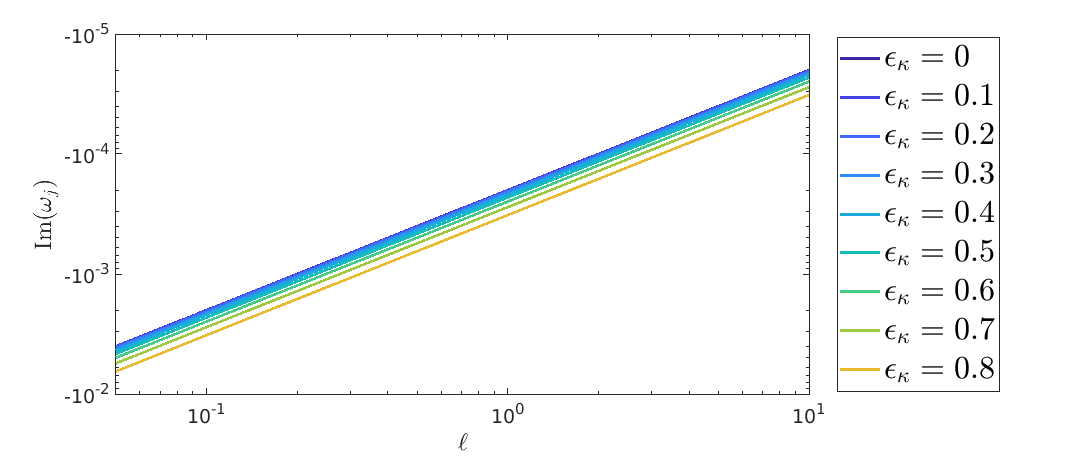}
    \caption{The imaginary parts of the non-zero quasifrequency $\omega_1$ for fixed $\delta=0.0001,\,v_{\mathrm{r}}=v_0=1,\,K=4,\,\Omega=0.03$ and time-modulated material parameters, as a function of $\ell$. Note that the quasifrequencies are purely imaginary. We use logarithmic axes to present our numerical results.}\label{fig:w_N1_timemod}
\end{figure}
The non-zero subwavelength quasifrequency $\omega_1\in\mathrm{i}\mathbb{R}$ shown in Figure \ref{fig:w_N1_timemod} was obtained with the capacitance matrix approximation derived in Theorem \ref{thm3.2}. It becomes apparent that for increasing $\ell$ the non-zero subwavelength quasifrequency $\omega_1(\delta)$ decreases, specifically,
\begin{align*}
    \omega_1(\delta)=O\left(\frac{1}{\ell}\right).
\end{align*} 
Furthermore, for increasing $\varepsilon_{\kappa}$ the non-zero subwavelength quasifrequency $\omega_1(\delta)$ decreases; see Figure \ref{fig:regimes_N1}. The left plot in Figure \ref{fig:regimes_N1} also serves as a numerical evidence of (\ref{eq:om1_Cepsk}).
\begin{figure}[H]
    \centering
    \includegraphics[width=\textwidth]{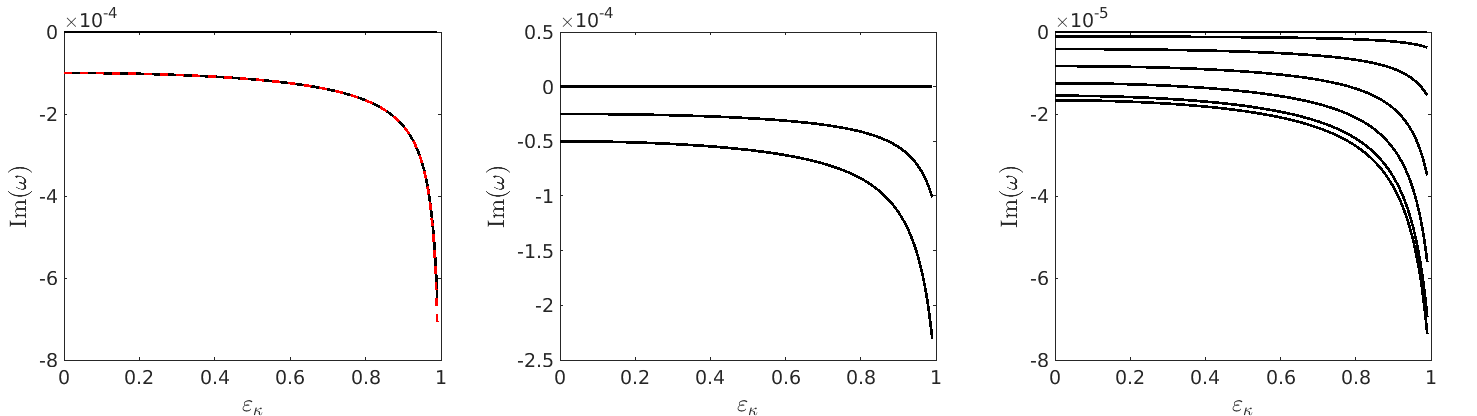}
    \caption{The imaginary parts of the subwavelength resonant quasifrequencies as a function of $\varepsilon_{\kappa}$ for $N\in\{1,2,6\}$ resonators each of length $\ell=2$ and spacing $\ell_{i(i+1)}=10$ at time $t=0$ with fixed $\delta=0.0001,\,v_{\mathrm{r}}=v_0=1,\,K=4,\,\Omega=0.03$. The red dashed line in the left plot shows the approximation computed via (\ref{eq:om1_Cepsk}).}\label{fig:regimes_N1}
\end{figure}
\section{Wave scattering from a system of time-modulated subwavelength resonators} \label{sec:chpt4}
We shall now derive a formulation of the scattered wave field upon a general incident wave field $u^{\inc}$. Firstly, we pose a system of linear equations defining the interior solution coefficients $\{(a_n^i,b_n^i)\}_{n\in \Z}$, for all $i=1,\dots,N$.
\begin{lemma}\label{lemma:AwF}
    The transmission condition in (\ref{eq:1DL_system}) can be expressed in terms of the following system of linear equations in $\mathbf{w}$:
    \begin{equation}\label{eq:lin_system}
        \mathcal{A}(\omega,\delta)\mathbf{w}=\delta \left(\mathbf{\mathcal{F}\left(\omega\right)}+\mathcal{N}\left(\omega\right)\mathbf{v}^{\inc}\right),
    \end{equation}
    where the doubly-infinite vector $\mathbf{\mathcal{F}\left(\omega\right)}$ is defined as
    \begin{equation*}
        \mathbf{\mathcal{F}}:= \left[\begin{smallmatrix}\svdots\\
            \mathbf{F}^{K}\\\svdots\\\mathbf{F}^{-K}\\\svdots
        \end{smallmatrix}\right], \quad 
        \mathbf{F}^j:=\begin{bmatrix}{F}^j_i\end{bmatrix}_{i=1,\dots,2N}, \; 
        {F}^j_i:=\begin{cases}
            \left.-\frac{\mathrm{d}v_j^{\inc}}{\mathrm{d}x}\right|_-(x_i^-),&i\,\, \text{odd},\\
            \left.\frac{\mathrm{d}v_j^{\inc}}{\mathrm{d}x}\right|_+(x_i^+),&i\,\, \text{even},
        \end{cases}
    \end{equation*}
    the doubly-infinite matrix $\mathcal{N}\left(\omega\right)$ is defined as
    \begin{align*}
        \mathcal{N}:=\mathrm{diag}\left(\mathcal{N}^{k^{(n)}}\right)_{n\in \Z},&\quad \mathcal{N}^{k^{(n)}}:=\left[\begin{smallmatrix}
            -\mathrm{i}k^{(n)}&&&&\\
            &N^{k^{(n)}}(\ell_{12})&&&\\
            &&\sddots&&\\
            &&&N^{k^{(n)}}(\ell_{(N-1)N})&\\
            &&&&-\mathrm{i}k^{(n)}
        \end{smallmatrix}\right],\\
        &N^{k^{(n)}}(\ell):=\begin{bmatrix}
            \frac{k^{(n)}\cos\left(k^{(n)}\ell\right)}{\sin\left(k^{(n)}\ell\right)} & \frac{-k^{(n)}}{\sin\left(k^{(n)}\ell\right)}\\
            \frac{k^{(n)}}{\sin\left(k^{(n)}\ell\right)} & \frac{-k^{(n)}\cos\left(k^{(n)}\ell\right)}{\sin\left(k^{(n)}\ell\right)}
        \end{bmatrix},
    \end{align*}
    and the doubly-infinite vector $\mathbf{v}^{in}$ is defined as
    \begin{align*}
        \mathbf{v}^{\inc}:=\left[\begin{smallmatrix} \svdots \\
            \mathbf{v}^{\inc}_K\\\svdots\\\mathbf{v}^{\inc}_{-K} \\ \svdots
        \end{smallmatrix}\right], \quad \mathbf{v}^{\inc}_j:=\begin{bmatrix}
            \left.v^{\inc}_n\right|_-(x_i^-)\\\left.v^{\inc}_n\right|_+(x_i^+)
        \end{bmatrix}_{1\leq i\leq N}.
    \end{align*}
    The doubly-infinite matrix $\mathcal{A}(\omega,\delta)$ is defined as in \cite[Theorem III.4]{jinghao_liora}. The unknown in (\ref{eq:lin_system}) is the vector of the coefficients corresponding to the interior solution, namely,
    \begin{equation*}
        \mathbf{w}:=\begin{bmatrix}\mathbf{w}_n\end{bmatrix}_{n\in \Z},\,\mathbf{w}_n:=\begin{bmatrix}a_n^i\\b_n^i\end{bmatrix}_{1\leq i\leq N}.
    \end{equation*}
\end{lemma}
\begin{proof}
    See the proof of \cite[Theorem III.4]{jinghao_liora}.
\end{proof}\par 

\begin{rem}
As pointed out in \Cref{rmk:K1}, we can numerically discretise \eqref{eq:lin_system} by truncating the Fourier series to length $K$, whereby the matrices and vectors become finite: 
$\mathcal{F}, \mathbf{v}^{\inc} \in\mathbb{C}^{2N(2K+1)}$, and $\mathcal{N}, \mathcal{A}\in\mathbb{C}^{2N(2K+1)\times2N(2K+1)}$.
\end{rem}
In contrast to the capacitance matrix formulation of \Cref{thm3.2}, which proved an asymptotic approximation of the subwavelength resonances $\omega_i$ in the limit $\delta\to 0$, these resonances can alternatively be approximated by finding the roots of the operator $\mathcal{A}(\omega,\delta)$ truncated at given $K$ large enough \cite{jinghao_liora}.  \par 
We now turn to the scattering problem where $u^\mathrm{in}(x,t)$ is non-zero. We denote the vector of the coefficients corresponding to the exterior solution by $\mathbf{v}$ and the vector of the coefficients corresponding to the interior solution by $\mathbf{w}$, in particular,
\begin{equation*} \mathbf{v}:=\begin{bmatrix}\mathbf{v}_n\end{bmatrix}_{n\in\mathbb{Z}},\,\mathbf{v}_n:=\begin{bmatrix}\alpha_n^i\\\beta_n^i\end{bmatrix}_{1\leq i\leq N},\quad \mathbf{w}:=\begin{bmatrix}\mathbf{w}_n\end{bmatrix}_{n\in\mathbb{Z}},\,\mathbf{w}_n:=\begin{bmatrix}a_n^i\\b_n^i\end{bmatrix}_{1\leq i\leq N}.
\end{equation*}
Lemma 2.1 in \cite{feppon_1d} provides a characterisation of the exterior coefficients for the case without an incident wave. Analogously to \cite[Lemma 2.1]{feppon_1d} we derive
\begin{align*}
    \begin{bmatrix}
        \alpha_n^i \\ \beta_n^i
    \end{bmatrix} = \frac{-1}{2\mathrm{i}\sin\left(k^{(n)}\ell_{i(i+1)}\right)}\begin{bmatrix}
        \mathrm{e}^{-\mathrm{i}k^{(n)}x_{i+1}^-} & -\mathrm{e}^{-\mathrm{i}k^{(n)}x_i^+} \\ -\mathrm{e}^{\mathrm{i}k^{(n)}x_{i+1}^-} & \mathrm{e}^{\mathrm{i}k^{(n)}x_i^+}
    \end{bmatrix}\begin{bmatrix}
        v^{\mathrm{sc}}_n(x_i^+) -v^{\inc}_n|_+(x_i^+)\\ v^{\mathrm{sc}}_n(x_{i+1}^-)-v^{\inc}_n|_-(x_{i+1}^-)
    \end{bmatrix}.
\end{align*}
Thus, we may introduce an operator $\mathcal{M}$ such that $\mathbf{v}=\mathcal{M}\mathbf{w}$. We define the operator $\mathcal{S}_n$ in $\mathbf{v}$ through
\begin{equation}\label{eq:Snv(x)}
    \mathcal{S}_n(\mathbf{v})(x):=\begin{cases} 
    \beta_n^0\mathrm{e}^{-\mathrm{i}k^{(n)}x}, & x \in (-\infty,x_1^-], \\
    \alpha_n^i\mathrm{e}^{\mathrm{i}k^{(n)}x}+\beta_n^i\mathrm{e}^{-\mathrm{i}k^{(n)}x}, & x \in [x_i^+,x_{i+1}^-], \\
    \sum\limits_{j=-\infty}^{\infty}\left(a_j^i\mathrm{e}^{\mathrm{i}\lambda_j^ix}+b_j^i\mathrm{e}^{-\mathrm{i}\lambda_j^ix}\right)f_n^{j,i}, & x\in(x_i^-,x_i^+), \\
    \alpha_n^N\mathrm{e}^{\mathrm{i}k^{(n)}x}, & x\in[x_N^+,+\infty).
    \end{cases}
\end{equation}
The coefficients $\alpha_n^N$ and $\beta_n^0$ are defined by
\begin{align}
    &\alpha_n^N:=\mathrm{e}^{-\mathrm{i}k^{(n)}x_N^+}\left(\sum\limits_{j=-\infty}^{\infty}\left(a_j^N\mathrm{e}^{\mathrm{i}\lambda_j^Nx_N^+}+b_j^N\mathrm{e}^{-\mathrm{i}\lambda_j^Nx_N^+}\right)f_n^{j,N}-v^{\inc}_n|_+(x_N^+)\right),\label{def:alpha_nN}\\ 
    &\beta_n^0:=\mathrm{e}^{\mathrm{i}k^{(n)}x_1^-}\left(\sum\limits_{j=-\infty}^{\infty}\left(a_j^1\mathrm{e}^{\mathrm{i}\lambda_j^1x_1^-}+b_j^1\mathrm{e}^{-\mathrm{i}\lambda_j^1x_1^-}\right)f_n^{j,1}-v_n^{\inc}|_-(x_1^-)\right).\label{def:beta_n0}
\end{align}
Inserting the definition of the operator $\mathcal{S}_n$ into the formula (\ref{def:us_Fourier}) defining the scattered wave field, we obtain
\begin{equation*}
    u^\mathrm{sc}(x,t)=\sum\limits_{j=1}^{N}\sum\limits_{n=-\infty}^{\infty}\mathcal{S}_n\left(\mathcal{M}\mathbf{w}\right)(x)\,\mathrm{e}^{-\mathrm{i}(\omega_j+n\Omega)t}.
\end{equation*}
Next, we note that by the definition of $\mathbf{w}$, it follows that
\begin{align*}
    \mathbf{w}=\left(\mathcal{A}(\omega,\delta)\right)^{-1}\left(\delta\mathbf{\mathcal{F}}+\delta\mathcal{N}\mathbf{v}^{\inc}\right).
\end{align*}
Assume that the operating frequency $\omega$ is close to the set of subwavelength quasifrequencies $\{ \omega_j\}_{j=1,\ldots N}$. Formally, we can then follow \cite[Proposition 4]{jde} and write in a small neighbourhood of the $\omega_j$'s
\begin{equation} \label{opexp}
    \left(\mathcal{A}(\omega,\delta)\right)^{-1}= \sum_{j=1}^N \frac{\mathcal{L}_j}{\omega-\omega_j}+\mathcal{R}(\omega),
\end{equation}
where \cite[Section 4.4]{jde}
\begin{equation*}
    \mathcal{L}_j=\frac{\langle\mathbf{w}^*,\cdot\rangle\mathbf{w}}{\big\langle\mathbf{w}^*,\frac{\mathrm{d}\mathcal{A}(\omega_j,\delta)}{\mathrm{d}\omega}\mathbf{w}\big\rangle},
\end{equation*}
and $\mathcal{R}(\omega)$ is a remainder term, which is holomorphic in $\omega$. The following proposition holds. 
\begin{prop} \label{propapprox}
    Neglecting the remainder $\mathcal{R}(\omega)$ in (\ref{opexp}),  we have the following approximation formula for the scattered field when the incident frequency $\omega$ is close to the set of subwavelength quasifrequencies 
    $\{\omega_j\}_{j=1,\dots,N}$, 
    \begin{align} \label{scatgeneral}
        u^\s(x,t) \approx \underbrace{\sum\limits_{j=1}^{N}\sum\limits_{n=-\infty}^{\infty}\mathcal{S}_n\left(\mathcal{M}\left(\frac{\mathcal{L}_j}{\omega-\omega_j}\left(\delta\mathbf{\mathcal{F}}+\delta\mathcal{N}\mathbf{v}^{\inc}\right)\right)\right)(x)\,\mathrm{e}^{-\mathrm{i}(\omega_j+n\Omega)t}}_{\ds=:\sum\limits^N_{j=1}\sum\limits^{\infty}_{n=-\infty}\frac{\gamma_n^j(x)}{\omega-\omega_j}\mathrm{e}^{-\mathrm{i}(\omega_j+n\Omega)t}},
    \end{align}
    where  $\gamma_n^j(x)$ are given by
    \begin{equation}\label{def:sc_1D}        
        \gamma_n^j(x):=\mathcal{S}_n\left(\mathcal{M}\left(\mathcal{L}_j\left(\delta\mathbf{\mathcal{F}}+\delta\mathcal{N}\mathbf{v}^{\inc}\right)\right)\right)(x).
    \end{equation}
\end{prop}
Thus, the total wave field is approximated through
\begin{align*}
    u(x,t)\approx u^{\mathrm{in}}(x,t)+\sum\limits^N_{j=1}\sum\limits^{\infty}_{n=-\infty}\frac{\gamma_n^j(x)}{\omega-\omega_j}\mathrm{e}^{-\mathrm{i}(\omega_j+n\Omega)t}.
\end{align*}
%\subsection{Definition of the incident wave field}
Note that so far we have not stated any further assumptions on $u^{\inc}$, except that it must satisfy (\ref{eq:PDE_uin}). Whence, we have not clarified whether the incident wave impinges on the resonators from the left or the right. One might consider now a left incident wave given by
\begin{align}\label{def:uin_l}
    u^{\inc}(x,t):=\begin{cases}
        \theta_1\mathrm{e}^{\mathrm{i}(kx-\omega t)},&x<x_1^-,\\
        0,&\text{otherwise},
    \end{cases}
\end{align}
a right incident wave given by
\begin{align*}
    u^{\inc}(x,t):=\begin{cases}
        \theta_2\mathrm{e}^{\mathrm{i}(-kx-\omega t)},&x>x_N^+,\\
        0,&\text{otherwise},
    \end{cases}
\end{align*}
or a left and right incident wave field given by
\begin{align*}
    u^{\inc}(x,t):=\begin{cases}
        \theta_1\mathrm{e}^{\mathrm{i}(kx-\omega t)},&x<x_1^- ,\\
        0,&x_1^-\leq x\leq x_N^+,\\
        \theta_2\mathrm{e}^{\mathrm{i}(-kx-\omega t)},&x>x_N^+.
    \end{cases}
\end{align*}
The constants $\theta_1$ and $\theta_2$ denote the amplitude of the incident wave field.\par  
In the remainder of this paper we shall consider the case of a left incident wave field, \textit{i.e.} an incident wave given by (\ref{def:uin_l}), with amplitude $\theta_1=1$. Upon substituting (\ref{def:uin_l}) into Lemma \ref{lemma:AwF}, we obtain the following corollary.
\begin{cor}
    We assume that an incident wave field impinges on a collection of $N$ resonators from the left. Then the coefficients
    \begin{align*}
        \mathbf{w}:=\begin{bmatrix}
            \mathbf{w}_n
        \end{bmatrix}_{n\in \Z},\,\mathbf{w}_n:=\begin{bmatrix}a_n^i\\b_n^i\end{bmatrix}_{1\leq i\leq N}
    \end{align*}
    of the interior solution are precisely the solution to the linear system
    \begin{align}\label{eq:AwFleft}
        \mathcal{A}(\omega,\delta)\mathbf{w}=\delta\mathcal{F}^{\mathrm{left}}(\omega),
    \end{align}
    where $\mathcal{F} ^{\mathrm{left}}(\omega)$ is now given by
    \begin{align*}
        \mathcal{F} ^{\mathrm{left}}:=\left[\begin{smallmatrix}\svdots\\
            \mathbf{F}^{\mathrm{left},K}\\\svdots\\\mathbf{F}^{\mathrm{left},-K}\\\svdots
        \end{smallmatrix}\right],\quad F^{\mathrm{left},j}_i:=\begin{cases}
            \left.-\frac{\mathrm{d}}{\mathrm{d}x}v_n^{\inc}\right|_-(x_i^-)-\left.\mathrm{i}k^{(n)}v^{\inc}_n\right|_-(x_i^-), &i=1,\\
            0,&i\,\,\text{otherwise}.
        \end{cases}
    \end{align*}
    The matrix $\mathcal{A}(\omega,\delta)$ is defined as in \cite[Theorem III.4]{jinghao_liora}.
\end{cor}
The linear system (\ref{eq:AwFleft}) can be solved numerically, for fixed $\delta$ and $\omega$, to obtain the total wave field solving (\ref{eq:1DL_system}) in the subwavelength regime. In Figure \ref{fig:usc_iterepsk}, we show the total wave field at time $t=0$ for various values of the modulation amplitude $\varepsilon_{\kappa}$. It becomes apparent from the numerical results in Figure \ref{fig:usc_iterepsk} that the total wave field's amplitude experiences a damping effect upon increasing $\varepsilon_{\kappa}$.
\begin{figure}[H]
    \centering
    \includegraphics[width=0.7\textwidth]{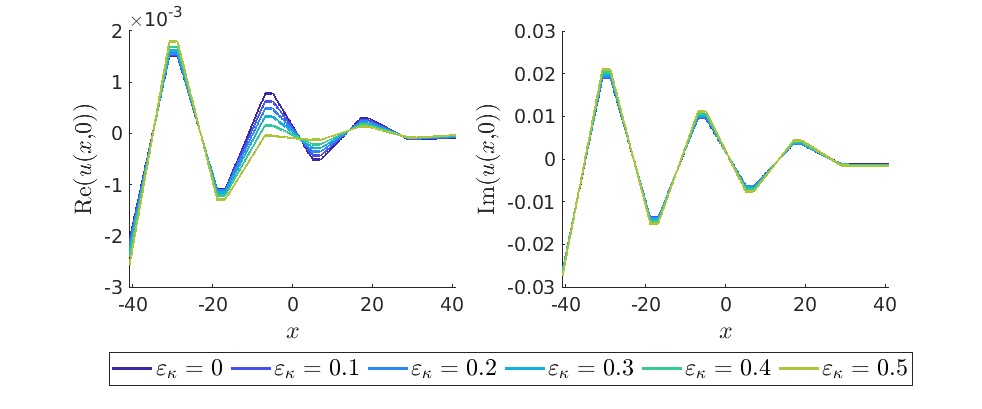}
    \caption{The scattered wave field at time $t=0$ for $N=6$ resonators each of length $\ell=2$ with a spacing of $\ell_{i(i+1)}=10$ between two neighbouring resonators. The operating frequency is taken to be very close to the subwavelength quasifrequency of (\ref{eq:new_CapApprox}) with the largest real part. The corresponding parameters are: $\delta=0.0001,\,v_{\mathrm{r}}=v_0=1,\,\Omega=0.03$.}\label{fig:usc_iterepsk}
\end{figure}

\begin{figure}[H]
    \begin{subfigure}{0.88\textwidth}
        \centering
        \includegraphics[width=0.9\textwidth]{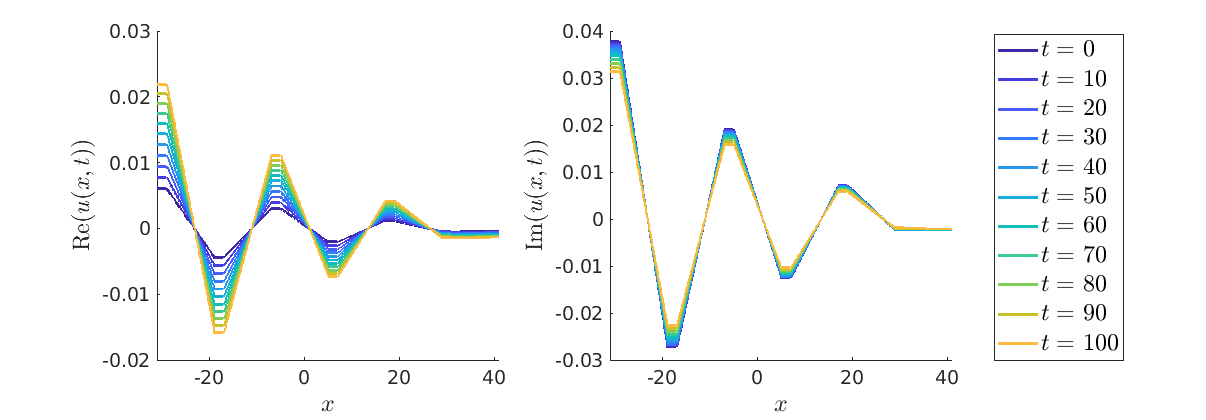}
        \caption{{\footnotesize Static material parameter, \textit{i.e.} $\varepsilon_{\kappa}=0$.}}
        \label{fig:t_evol_static}
    \end{subfigure}
    \centering
    \begin{subfigure}{0.88\textwidth}
        \centering
        \includegraphics[width=0.9\textwidth]{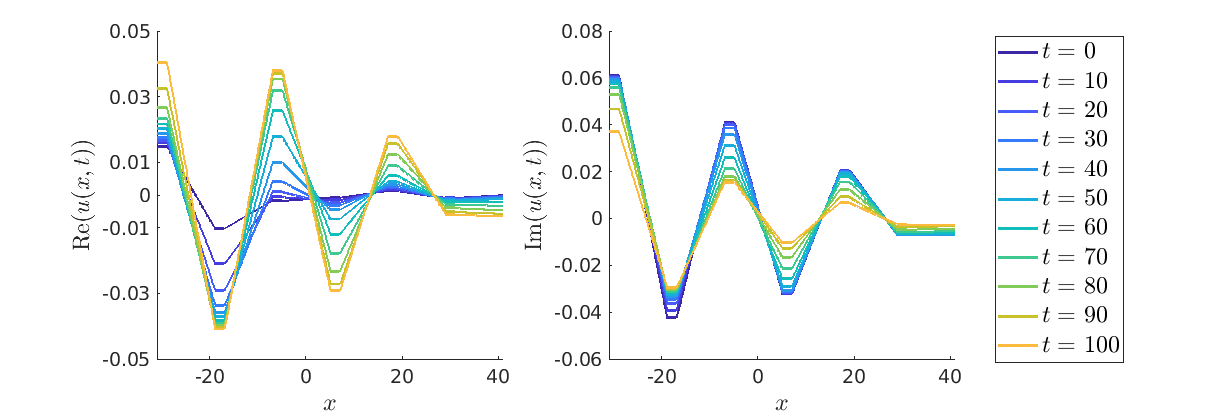}
        \caption{{\footnotesize Time-modulated material parameter with $\varepsilon_{\kappa}=0.9$.}}
        \label{fig:t_evol_tdep}
    \end{subfigure}
    \caption{The real and imaginary parts of the total wave field measured at several times $t$. These numerical results were obtained for $N=6$ resonators each of length $\ell=2$ and spacing $\ell_{i(i+1)}=10$. The operating frequency is taken to be very close to the subwavelength quasifrequency of (\ref{eq:new_CapApprox}) with the largest real part. The corresponding parameters are: $\delta=0.0001,\,v_{\mathrm{r}}=v_0=1,\,\Omega=0.03$.}
    \label{fig:t_evol}
\end{figure}
Firstly, we aim to understand how the time evolution of the total wave field is manipulated by time-modulated material parameters. Comparing the results in Figure \ref{fig:usc_iterepsk} and Figure \ref{fig:t_evol},  it becomes apparent that the amplitude of the total wave field can be manipulated through the amplitude $\varepsilon_{\kappa}$ of the material parameter $\kappa(x,t)$ in a similar way as it would evolve over time. Thus, to obtain a specific amplitude it is more efficient to do that by modulating the material parameter $\kappa$ over time rather than letting the wave propagate until this amplitude is achieved.\par

\section{Scattering coefficients and energy conservation}\label{sec:chpt5}
We shall now solve the governing equations numerically by solving the linear system (\ref{eq:AwFleft}) for the interior coefficients of the solution. Specifically, we want to understand which phenomena can be provoked by time-dependent material parameters. In the remainder of this section, we assume that $\kappa_i(t)$ is given by (\ref{eq:rho_kappa}), but with uniform amplitudes, \textit{i.e.} $\varepsilon_{\kappa,i}=\varepsilon_{\kappa}$ for all $i=1,\dots,N$. We consistently define the phase shifts to be
\begin{align*}
    \phi_{\kappa,i}:=\frac{\pi}{i}.
\end{align*}
Throughout this section, the wave speeds inside and outside of $D$ are assumed to be $v_{\mathrm{r}}=v_0=1$. Moreover, the contrast parameter is $\delta=0.0001$, the modulation frequency is $\Omega=0.03$ and the truncation parameter is set to be $K=4$. Recall that we assume the incident wave field to impinge on the system of resonators from the left-hand side, \textit{i.e.} $u^{\inc}$ is given by (\ref{def:uin_l}), with $\theta_1=1$. \par  
Previous work has shown that one-dimensional systems of static subwavelength resonators obey energy conservation \cite{feppon_1d}. For any system of 
time-modulated subwavelength resonators, we can define reflection and transmission coefficients of the $n$-th mode as follows.
\begin{definition}\label{def:TR}
    The reflection and transmission coefficients of the $n$-th scattered mode of a system of $N$ time-modulated subwavelength resonators $D_i:=(x_i^-,x_i^+)$ are defined by
    \begin{align}
        R_n:=\beta_n^0,\, T_n:=\alpha_n^N,\quad\forall\,n\in\mathbb{Z},
    \end{align}
    respectively. The coefficients $\alpha_n^N$ and $\beta_n^0$ are defined by (\ref{def:alpha_nN}) and (\ref{def:beta_n0}), respectively.
\end{definition}
It is a well-established fact that in the static case, the system must conserve its energy, that is, \cite[Proposition 4.4]{feppon_1d}
\begin{align}\label{eq:staticE}
    |R_0|^2+|T_0|^2=1.
\end{align}
Using the reflection and transmission coefficients we can introduce the notion of scattered energy flux of a system of time-modulated subwavelength resonators. The value of $|R_n|^2+|T_n|^2$ can be seen as the so-called scattering cross-section of the $n$-th mode, which can be understood as the scattering rate of the mode per unit incident flux \cite{lindon}.
\begin{definition}\label{def:totEnergy}
    The scattered energy flux of a system of time-modulated subwavelength resonators is defined by
    \begin{align*}
        E := \sum_{n=-\infty}^\infty |T_n|^2 + |R_n|^2.
    \end{align*}
    We say that a system 
    \begin{itemize}
        \item conserves energy if $E=1$;
        \item loses energy if $E<1$;
        \item gains energy if $E>1$.
    \end{itemize}
\end{definition}
The hereby introduced notion of total energy $E$ agrees with the optical theorem \cite[Theorem 2.10.4.3]{mcmpp}.
\begin{rem}
    Note that the reference value $1$ in Definition \ref{def:totEnergy} comes from the incident flux $(\theta_1)^2=1$.
\end{rem}
\begin{figure}[H]
    \centering
    \includegraphics[width=0.6\textwidth]{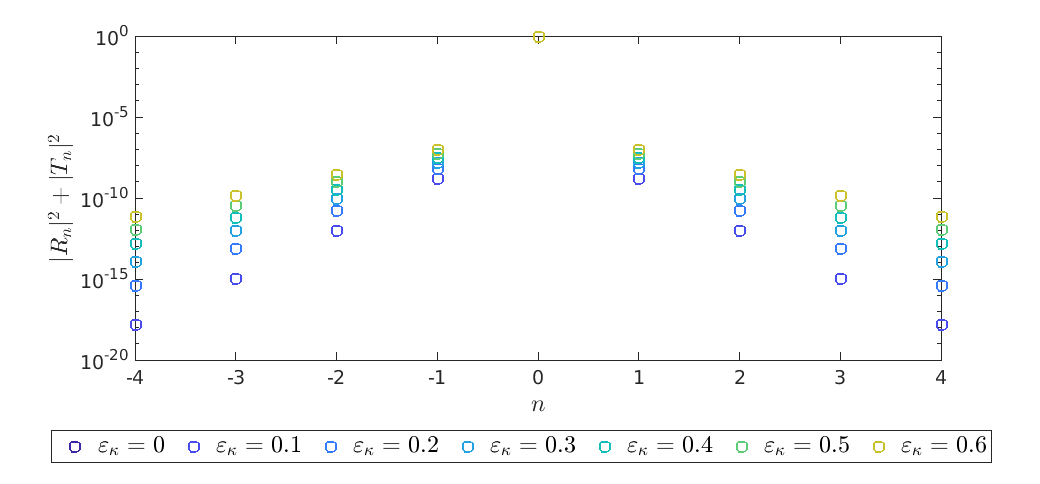}
    \caption{System of $N=6$ resonators of length $\ell_i=2$ and spacing $\ell_{i(i+1)}=10$. The operating frequency $\omega$ is taken to be very close to the subwavelength quasifrequency of (\ref{eq:new_CapApprox}) with the largest real part. Here, the time-modulation gives rise to frequency conversion, whereby the scattered field consists of a superposition of fields of multiple frequencies. }\label{fig:energycons_epskappa}
\end{figure}\par 
Figure \ref{fig:energycons_epskappa} shows the energy distribution over the different modes of the scattered field. It indicates that the energy decays exponentially as a function of the mode order $n$. In this case, the incident field is entirely supported in the mode $n=0$. \par 
Opposed to the static case, under time-modulation of the material parameter $\kappa$, our system does not conserve its energy anymore, as shown in Figure \ref{fig:energycons_epskappa}. How much energy a system loses or gains is affected by the amplitude of modulation $\varepsilon_{\kappa}$ and the operating frequency $\omega$. We can numerically determine the following three regimes:
\begin{enumerate}
    \item[(R1)] $\varepsilon_{\kappa}$ and $\omega$ such that $E>1$ (energy gain);
    \item[(R2)] $\varepsilon_{\kappa}$ and $\omega$ such that $E=1$ (energy conservation);
    \item[(R3)] $\varepsilon_{\kappa}$ and $\omega$ such that $E<1$ (energy loss).
\end{enumerate}
Figure \ref{fig:Et_N126} shows the total energy as a function of $\varepsilon_{\kappa}$ and Figure \ref{fig:Et_N126_omega} shows the total energy as a function of $\omega$. Note that it is not possible to conclude from Figure \ref{fig:usc_iterepsk}, which shows the total wave field at $t=0$ for various values of $\varepsilon_{\kappa}$, whether a system loses or gains energy. Thus, the only way to make a qualitative statement about the system's energy is by considering the value of the total energy $E$, such as in Figures \ref{fig:Et_N126} and \ref{fig:Et_N126_omega}.\par
The numerical results in Figure \ref{fig:Et_N126_omega} show that operating frequencies near the subwavelength quasifrequencies result in a greater energy gain when the material parameter $\kappa$ is time-modulated. We see that in the herewith considered settings, time-modulation of $\kappa$ results in an energy gain, which can be amplified by choosing $\omega$ close to a resonant quasifrequency.
\begin{figure}[H]
    \centering
    \includegraphics[width=\textwidth]{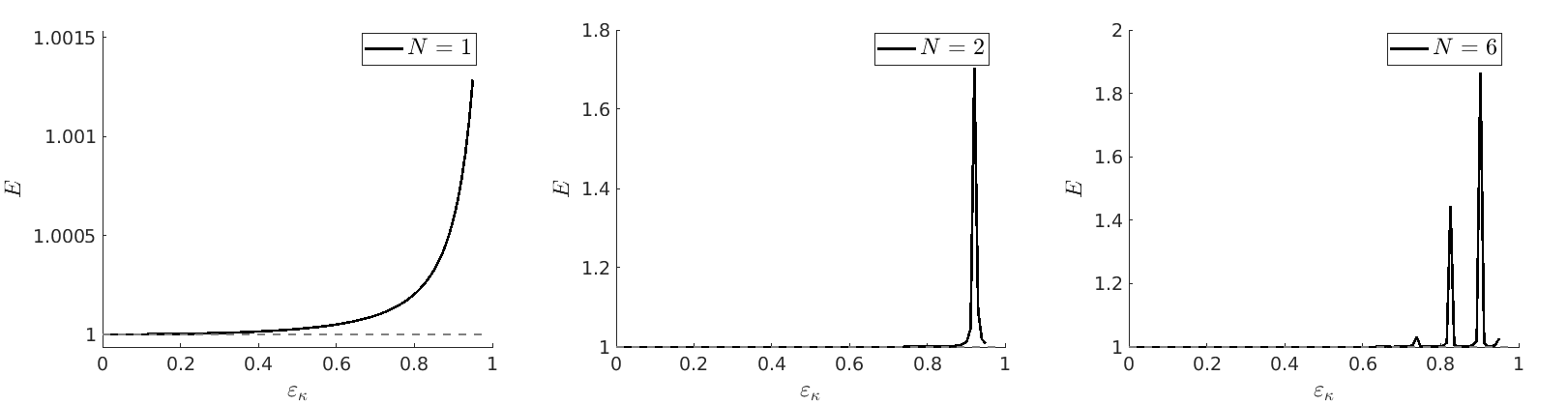}
    \caption{The total energy $E$ for $N\in\{1,2,6\}$ resonators each of length $\ell=2$ and spacing $\ell_{i(i+1)}=10$. The operating frequency is chosen to be close to the subwavelength quasifrequency with the largest real part. The grey dashed line marks the value $(\theta_1)^2=1$ as a reference.}\label{fig:Et_N126}
\end{figure}
\begin{figure}[H]
    \centering
    \includegraphics[width=\textwidth]{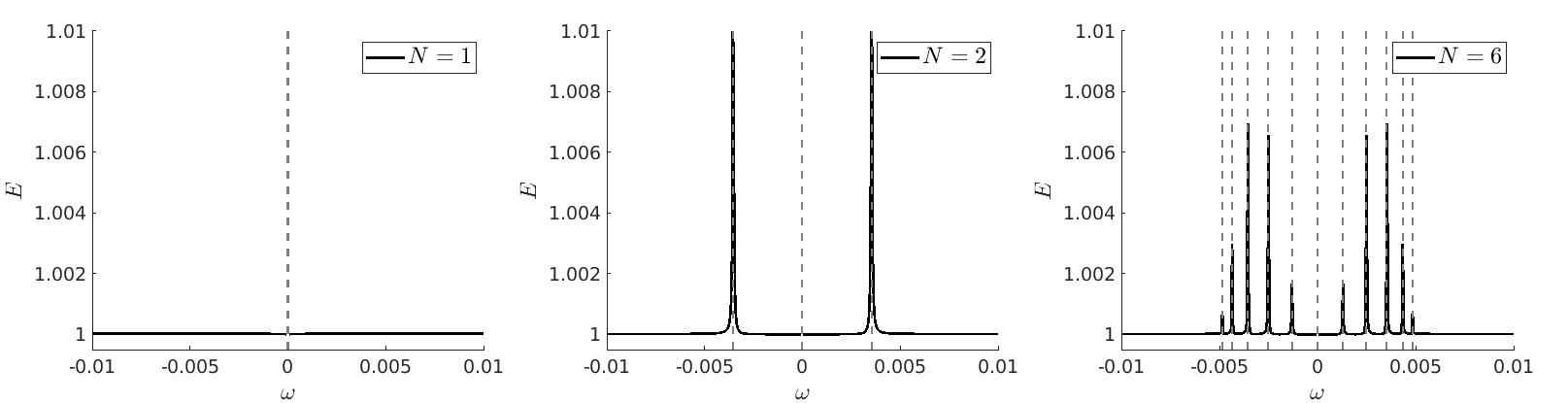}
    \caption{The total energy $E$ for $N\in\{1,2,6\}$ resonators each of length $\ell=2$ and spacing $\ell_{i(i+1)}=10$ with time-modulated material parameter $\kappa$ with amplitude $\varepsilon_{\kappa}=0.6$. The grey dashed line marks the real parts of the subwavelength quasifrequencies $\omega_i$.}\label{fig:Et_N126_omega}
\end{figure}

\section{Concluding remarks}\label{sec:chpt6}
We have derived an explicit formulation of the scattered wave field arising from a system of subwavelength resonators with time-modulated material parameters. This was done by using a pole-pencil decomposition over the subwavelength resonant modes. Additionally, we have provided a numerical scheme to solve the governing equations (\ref{eq:1DL_system}) by rewriting them as a system of linear equations (\ref{eq:lin_system}) whose unknowns are the coefficients of the interior solution. In this work we have considered an incident wave impinging on the resonators from the left-hand side. However, our framework allows any definition for the incident wave field. Contrary to the higher-dimensional case, only the first resonator $D_1$ sees the left-hand side incident wave field. This explains why the vector on the right-hand side of (\ref{eq:AwFleft}) is sparse.\par 
In order to allow for further numerical experiments involving the subwavelength quasifrequencies, we have proved a higher-order, discrete, capacitance matrix approximation for a system resonators in Theorem \ref{thm3.2}. It is worth emphasising that our formula approximates the subwavelength quasifrequencies up to order $O(\delta^{3/2})$, as opposed to $O(\delta)$, formerly derived in \cite[Theorem V.3]{jinghao_liora}. Moreover, we derived a specific formulation of the non-zero subwavelength resonant frequency of a single-resonator system in Corollary \ref{cor:N1}, which can be seen as a generalisation of \cite[Remark 3.4]{feppon_1d} to time-modulated systems.\par 
Opposed to the static case, our numerical results clearly show that upon modulating the material parameter $\kappa$ in time, the energy of a system of subwavelength resonators is not being conserved. We have introduced a new notion of total energy of a system of $N$ subwavelength resonators with time-modulated material parameter $\kappa$ in Definition \ref{def:totEnergy}. The total energy depends on the reflection and transmission coefficients $R_n$ and $T_n$, respectively, for each mode. We have showed how the energy behaves as a function of the amplitude of the modulation $\varepsilon_{\kappa}$ and the operating frequency $\omega$. Note that the ratios between $|R_n|^2+|T_n|^2$ and $\sum_n |R_n|^2+|T_n|^2$  tell us the distribution of the energy over the different modes. \par 
We conclude that this paper gives a novel formulation of the scattered wave field upon an incident wave field, and it presents new key results regarding the total energy of a time-modulated system of subwavelength resonators. Based on the results of this paper, it would be very interesting to derive a novel effective medium theory for the time-modulated setting, which holds true in the dilute regime, as the number of resonators goes to infinity. The generalisation of the point approximation method to the time-modulated setting and its numerical implementation would be the subject of a forthcoming paper. 

\appendix
\section{Code availability}
The codes that were used to generate the results presented in this paper are available under \url{https://github.com/rueffl/Scattering_Problem}.

\addcontentsline{toc}{chapter}{References}
\renewcommand{\bibname}{References}
\bibliography{refs}
\bibliographystyle{plain}

\end{document}